 \documentclass[10pt,journal,final]{IEEEtran}
\usepackage[comma,numbers,square,sort&compress]{natbib}
\usepackage{amsmath}
\usepackage{amssymb}
\usepackage{float}
\usepackage{amsthm}
\usepackage{graphicx}
\usepackage{epstopdf}
\usepackage{color}
\usepackage{verbatim}
\usepackage{tikz,times}
\usepackage{algorithm}
\usepackage{algorithmic}
\usepackage{diagbox}
 \DeclareMathOperator{\blockdiag}{blockdiag}

  \DeclareMathOperator{\Expect}{\mathbb{E}}
   \DeclareMathOperator{\Pro}{\mathbb{P}}
 \newcommand{\norm}[1]{\left\lVert#1\right\rVert}

\theoremstyle{plain}
\newtheorem{theorem}{Theorem}[section]

\newtheorem{lemma}{\textbf{Lemma}}[section]

\theoremstyle{definition}
\newtheorem{definition}{\textbf{Definition}}[section]
\newtheorem{assumption}{\textbf{Assumption}}[section]

\theoremstyle{remark}
\newtheorem{remark}{Remark}[section]
\begin{document}

\title{ Distributed Parameter   Estimation Under Event-triggered Communications}

\author{Xingkang~He, Qian Liu, Junfeng Wu, 	Karl Henrik Johansson
\thanks{The work is supported by Knut \& Alice Wallenberg foundation, and by Swedish Research Council.}
\thanks{Xingkang He and Karl Henrik Johansson are with the ACCESS Linnaeus Centre, School of Electrical Engineering and Computer Science. KTH Royal Institute of Technology, Sweden (xingkang@kth.se, kallej@kth.se).}
\thanks{ Qian Liu is with LSC, Academy of Mathematics and Systems Science, Chinese Academy of Sciences, China; and she is also with University of Chinese Academy of Sciences, China (liuqian@amss.ac.cn).}
\thanks{Junfeng Wu is with  College of Control Science and Engineering. Zhejiang University, China (jfwu@zju.edu.cn).}
}

\maketitle

\begin{abstract}
In this paper, we study a distributed parameter estimation problem with an asynchronous communication protocol over multi-agent systems.
Different from   traditional time-driven communication schemes,  in this work, data can be transmitted between agents intermittently 
  rather than in a steady stream. 
First, we propose a recursive distributed estimator based on an event-triggered communication scheme, through which each agent can decide whether the current estimate is sent out to its neighbors or not. With this scheme, considerable communications between agents can be effectively reduced.
Then, under mild conditions including a collective observability, we provide a design principle of triggering thresholds to guarantee the asymptotic  unbiasedness and strong consistency.  Furthermore, under certain conditions, we prove that, with probability one, for every agent the time interval between two successive triggered instants goes to infinity as time goes to infinity. Finally, we provide a numerical simulation to validate the theoretical results of this paper.
\end{abstract}


\section{Introduction}
As one of the hottest research topics over the last decade, multi-agent systems have attracted a  lot of attention of researchers around the world due to their broad applications in sensor networks, cyber-physical systems, computer games, transportation, etc. With the development of network technology and the increasing of data amount, distributed learning and estimation protocols without requiring a data center are becoming more and more popular.

Distributed parameter estimation over multi-agent systems is on the problem of  estimation or learning of an unknown parameter based on  data transmission between neighboring agents. Numerous practical applications, such as temperature monitoring, weather prediction and environmental exploration, can be cast into distributed parameter estimation problems. Due to   environmental complexity, the estimation problem is usually modeled under  stochastic frameworks, where  measurements of each agent are polluted by random noises. In \cite{rad2010distributed,kar2011convergence,kar2013distributed,cattivelli2010diffusion}, the distributed parameter estimation problems are investigated with respect to the estimation properties including consistency and asymptotic normality.
The distributed parameter estimation problems over random networks and imperfect communication channels are studied in  \cite{zhang2012distributed,kar2012distributed}. The connection between  graph topologies and  estimation performance in terms of asymptotic variances are  effectively analyzed in above literature.

Design and analysis of communicaton schemes between agents is an essential research topic of networked estimation and control. Due to the limitations of channel capacity and energy resources, traditional time-driven communication schemes may not be suitable to some practical applications, such as wireless agent networks.
Thus, in the existing literature, there have been a  few results which consider   event-triggered communication schemes. 
Event-triggered measurement scheduling problems  are well studied in \cite{you2013asymptotically,shi2014event,mo2012kalman,sinopoli2004kalman,han2015optimal,han2015stochastic,weimer2012distributed}. In these literature, the parameter estimation or filtering problems are investigated under  centralized frameworks, where a center can process the transmitted messages to obtain estimates for parameter vector or state vector. \cite{He2017On,battistelli2018distributed,he2017distributed} study   distributed  filtering problems with  event-triggered communications, where the messages of state estimates or covariance bounds are transmitted to other agents intermittently. 
However, to the best knowledge of authors, the distributed parameter estimation problems  with event-triggered communications have not been well studied in the existing literature. The main difficulty is to design and analyze triggering conditions so as to reduce  communication frequency between agents with guaranteed  estimation properties.

In this paper, we study the distributed parameter estimation problem with event-triggered intermittent communications between neighboring agents.
The contributions of this work are two fold. First, we propose  an event-triggered communication scheme, through which each agent can decide whether the current estimate is sent out to its neighbors or not.   With this scheme, redundant communications between agents can be effectively reduced. Second, under mild conditions, for the considered distributed estimator, we prove the main estimation properties including asymptotic unbiasedness and strong consistency. Besides, 
we prove that,   for every agent the time interval between two successive triggered instants goes to infinity as time goes to infinity
  in the sense of almost sure, which means the communication frequency between any two neighboring agents is tremendously reduced if time is sufficiently large.  
  It should be noted that the main difference between the event-triggered framework proposed in this work and the existing literature is that our triggering threshold will go to zero as time goes to infinity, which is necessary to guarantee the asymptotic convergence of estimates in mild collective observability conditions.

The remainder of the paper is organized as follows: Section 2 is on preliminaries  and problem formulation. Section 3 considers the event-triggered communication scheme and  some  main asymptotic estimation properties.  Section 4 provides a numerical simulation. The conclusion of this paper is given in Section 5.

\subsection{Notations}
 The superscript ``T" represents the transpose.  $I_{n}$ stands for the $n$-dimensional square identity matrix. 
 $\textbf{1}_N$ stands for the $N$-dimensional vector with all elements being one. 
 $\Expect\{x\}$ denotes the mathematical expectation of the stochastic variable $x$, and  $\blockdiag\{\cdot\}$   represent the diagonalizations of block elements.  Additionally, $i.i.d.$ is the abbreviation of `independent identically distribution'.
$A\otimes B$ is the Kronecker product of $A$ and $B$.  $\norm{x}$ is the norm of a vector $x$.  The mentioned scalars, vectors and matrices of this paper are all real-valued. $\mathbb{A}^{n\times m}$ is the real matrix with $n$ rows and $m$ columns. w.r.t. is the abbreviation of `with respect to'.

\section{Preliminaries and Problem Formulation}
In this section, we provide some necessary graph preliminaries and then formulate the problem studied in this work.
\subsection{Graph Preliminaries}
In this paper, the communication between agents of a network is modeled as an undirected graph $\mathcal{G=(V,E,A)}$, which consists of  the set of  nodes $\mathcal{V}$, the set of edges $\mathcal{E}\subseteq \mathcal{V}\times \mathcal{V}$, and the adjacency matrix $\mathcal{A}=[a_{i,j}]$.  $\mathcal{A}$ is a symmetry matrix consisting of one and zero.
If $a_{i,j}=1,j\neq i$, there is an edge $(i,j)\in \mathcal{E}$, which means node $i$ can exchange information with node $j$, and node $j$ is called a  neighbor of node $i$. For node $i$, the neighbor set of agent $i$ is denoted by $\{j\in\mathcal{V}|(i,j)\in \mathcal{E}\}\triangleq\mathcal{N}_{i}$.
We suppose that the graph has no self-loop, which means $a_{i,i}=0$ for any $i\in\mathcal{V}$.
$\mathcal{G}$ is called connected if for any pair nodes $(i_{1},i_{l})$, there exists a   path from $i_{1}$ to $i_{l}$ consisting of edges $(i_{1},i_{2}),(i_{2},i_{3}),\ldots,(i_{l-1},i_{l})$. Besides, we denote $\mathcal{L}=\mathcal{D}-\mathcal{A}$, where $\mathcal{L}$ is called Laplacian matrix and $\mathcal{D}$ called degree matrix. $\mathcal{D}$ is a diagonal matrix consisting of numbers of neighboring nodes. For detailed definitions, the readers are referred to \cite{Mesbahi2010Graph}.
On the connectivity of a graph, the following theorem holds.
\begin{theorem}\cite{Mesbahi2010Graph}
	The graph $\mathcal{G}$ is connected if and only if $\lambda_{2}(\mathcal{L})>0$.
\end{theorem}

\subsection{Problem Setup}
Consider the unknown parameter vector $\theta\in \mathbb{R}^n$ is observed by $N>0$ agents with the following model
\begin{equation}\label{system_all}
y_{i}(t)= H_{i}\theta+v_{i}(t),i=1,2,\cdots,N,
\end{equation}
where 
 $y_{i}(t)\in \mathbb{R}^{m_{i}}$ is the measurement vector, $v_{i}(t)\in \mathbb{R}^{m_{i}}$  is the zero-mean $i.i.d.$ measurement noise with covariance $R_{i}$, and $H_{i}\in \mathbb{R}^{m_{i}\times n}$  represents the known   measurement matrix of agent $i$. The noise covariance matrix of all agents is $R_v$, where $R_v=\Expect\{V(t)V(t)^T\}$ and $V(t)=\left[v_1^T(t),\dots,v_N^T(t)\right]^T$. Note that we simply require the temporal independence of measurement noises, thus the noises of agents could be spatially correlated.

Assume $x_{i}(t)$ is the  estimate of agent $i$ at time $t$ for the parameter vector $\theta$.
In \cite{kar2011convergence}, the  following estimator is studied
\begin{align}\label{eq_estimator0}
x_{i}(t+1)=&x_{i}(t)+\beta(t)\sum_{j\in\mathcal{N}_{i}}(x_{j}(t)-x_{i}(t))\\
&+\alpha(t) KH_{i}^T( y_{i}(t)-H_{i}x_{i}(t)),\nonumber
\end{align}
where $\alpha(t)$ and $\beta(t)$ are time-varying steps satisfying certain conditions. $K$ is the parameter to be designed.

To reduce   limitation of energy consumption and alleviate burden of communication channels, we focus on studying event-triggered communication scheme, in which the state estimates of each agent will not be consistently transmitted.  

In the following of this paper, we focus on solving the problems as follows.
1) How to design a fully distributed event-triggered communication scheme for each agent?
2) What conditions are required to guarantee essential estimation properties, including asymptotic unbiasedness and strong consistency for each agent?
3) How does the event-triggered communication scheme  contribute to reducing the communication frequency of agents with guaranteed properties?

\section{Main results}
In this section, we will propose an event-triggered communication scheme and analyze the main estimation performance of a recursive distributed estimator based on the triggering scheme.
\subsection{Event-triggered communication scheme}

In this subsection, we consider design an event-triggered scheme, which can decide for each agent whether the current estimate is sent out to its neighbors or not.
Let  $t_{k}^i$ be the $k$th triggering time of the $i$th agent, and it is the latest triggering time of agent $i$. 
Then, we define the  triggering event $\mathbf{E}_{i}(t)$ 
\begin{align}\label{eq_event}
\mathbf{E}_{i}(t): \norm{x_{i}(t)-x_{i}(t_{k}^i)}>\frac{1}{ (t+1)^{\rho_i}}, i\in\mathcal{V},
\end{align}
where $\rho_i$ is a positive scalar addressed in the following, and $x_{i}(t_{k}^i)$ is the latest state estimate sent out by agent $i$ at time $t_{k}^i$.

Let $\rho_0=\min\{\rho_j,\forall j\in\mathcal{V}\}$, and the following  random indicator variable $\gamma_i(t)$ be
\begin{align}\label{eq_trigger_scheme}
\gamma_i(t)=\begin{cases}
0, \text{ if } \mathbb{E}_{i}(t) \text{ occurs},\\
1, \text{ otherwise}.
\end{cases}
\end{align}
Note that the distribution of $\gamma_i(t)$ influences the communication frequency of the whole multi-agent system. If $\Pro(\gamma_i(t)=0)=0$, $\forall i\in\mathcal{V}$, $t\in\mathbb{N}$, then the communications between agents will not happen almost surely. And if $\Pro(\gamma_i(t)=0)=1$, $\forall i\in\mathcal{V}$, $t\in\mathbb{N}$, the communication scheme is equivalent  to the time-driven one almost surely.

Based on the triggering scheme (\ref{eq_event}) - (\ref{eq_trigger_scheme}),  for agent $i$, we propose the following event-triggered distributed parameter estimator
\begin{align}\label{eq_estimator}
&x_{i}(t+1)=x_{i}(t)+\alpha(t)  H_{i}^T( y_{i}(t)-H_{i}x_{i}(t)),\\
&\quad+\beta(t)\sum_{j\in\mathcal{N}_{i}}\left(\gamma_i(t)x_{j}(t_{k}^j)+(1-\gamma_i(t))x_{j}(t)-x_{i}(t)\right),\nonumber
\end{align}
where $\alpha(t)$ and $\beta(t)$ are time-varying steps addressed in  Assumption \ref{ass_steps}.

\begin{remark}
	To achieve the estimator (\ref{eq_estimator}), each agent should reserve the latest state estimates which its neighboring agents sent out. If new estimates come, the stored ones can be updated. 
\end{remark}
\begin{remark}
Different from existing results \cite{you2013asymptotically,shi2014event,mo2012kalman,sinopoli2004kalman,han2015optimal,han2015stochastic,weimer2012distributed,He2017On,battistelli2018distributed,he2017distributed}, the triggering threshold $\frac{1}{ (t+1)^{\rho_i}}$ goes to zero as $t$ goes to infinity. If the threshold does not go to zero as time goes to infinity and the collective observability condition (see Assumption \ref{ass_obser}) without agent $i$ is not satisfied, the estimates of all agents except agent $i$ will not converge to the true parameter. 
\end{remark}

%

\subsection{Performance Analysis}
For convenience, we provide the following  notations.
\begin{align}\label{eq_denotations}
\begin{cases}
\Theta=\textbf{1}_N\otimes \theta\\
Y(t)=\left[y_1^T(t),\dots,y_N^T(t)\right]^T\\
X(t)=\left[x_1^T(t),\dots,x_N^T(t)\right]^T\\
X(t_k)=\left[x_{1}^T(t_{k}^1),\dots,x_{N}^T(t_{k}^N)\right]^T\\
\bar D_H=\blockdiag\left\{H_1^T,\dots,H_N^T\right\}\\
 D_H=\blockdiag\left\{H_1^TH_1,\dots,H_N^TH_N\right\}.
\end{cases}
\end{align}

%
%

The following assumptions are needed  in this paper.
\begin{assumption}\label{ass_graph}
The  graph is connected, i.e., $\lambda_2(\mathcal{L})>0$.
\end{assumption}
\begin{assumption}\label{ass_obser}
	The observation system (\ref{system_all}) is collectively observable, i.e., 
	$G=\sum_{i=1}^{N}H_i^TH_i$ is full rank. 
\end{assumption}
\begin{assumption}\label{ass_noise}
	There exists a positive scalar $\epsilon_1>0$, such that $\Expect \{\norm{V(t)}^{2+\epsilon_1}\}<\infty.$
\end{assumption}
\begin{assumption}\label{ass_steps}
	The steps in (\ref{eq_estimator}) are set with  $\alpha(t)=\frac{a}{(t+1)^{\tau_1}}$ and $\beta(t)=\frac{b}{(t+1)^{\tau_2}}$, where $a,b>0$, $0<\tau_2\leq\tau_1\leq 1$. Besides, $\tau_1>\max\{\tau_2+\frac{1}{2+\epsilon_1},0.5\}$.
\end{assumption}
\begin{remark}
Assumption \ref{ass_graph} is a common condition of distributed estimation and control for multi-agent systems. Assumption \ref{ass_obser} is a collective observability condition, which is satisfied even if any local observability condition is not satisfied. Assumption \ref{ass_noise} is on the moment condition of noises, which requires a little severe than boundedness of mean square. Assumption \ref{ass_steps} provides feasible design conditions of the steps in (\ref{eq_estimator}). 
\end{remark}

On the triggering scheme in (\ref{eq_trigger_scheme}), if $\gamma_i(t)=0$, the agent $i$ will send its estimate $x_{i}(t)$ to its neighbors, who then update the stored estimate with $x_{i}(t_{k+1}^i)=x_{i}(t)$. Thus,   we rewrite (\ref{eq_estimator}) in the following form
\begin{align}\label{eq_estimator02}
x_{i}(t+1)=&x_{i}(t)+\beta(t)\sum_{j\in\mathcal{N}_{i}}(x_{j}(t)-x_{i}(t))\nonumber\\
&+\alpha(t)H_{i}^T( y_{i}(t)-H_{i}x_{i}(t))\\
&+\beta(t)\sum_{j\in\mathcal{N}_{i}}(x_{j}(t_{k}^j)-x_{j}(t))\nonumber.
\end{align} 
Considering  the notations in (\ref{eq_denotations}), 
we  obtain the compact form of (\ref{eq_estimator02}) in the following
\begin{align}\label{eq_estimator03}
X(t+1)=&X(t)-\beta(t)(\mathcal{L}\otimes I_M)X(t)\nonumber\\
&+\alpha(t)\bar D_H(Y(t)-\bar D_H^TX(t))\\
&+\beta(t)(\mathcal{A}\otimes I_M)(X(t_{k})-X(t))\nonumber.
\end{align}

We have the following lemma on the error between transmitted estimate vector $X(t_{k})$ and current estimate vector $X(t)$.
\begin{lemma}\label{eq_error_trans}
	Consider (\ref{eq_estimator03}), then there exists a scalar $\bar m>0$, such that
	\begin{align}
\norm{X(t_{k})-X(t)}\leq  \frac{\bar m}{(t+1)^{\rho_0}}.
	\end{align}
\end{lemma}
The following two lemmas are useful to further analysis.
\begin{lemma}\label{lem_posit}
	Under Assumption \ref{ass_graph} and Assumption \ref{ass_obser}, $\mathcal{L}\otimes I_M+D_H$ is a positive definite symmetry matrix. Furthermore, there exists a  constant matrix $M_0\in\mathbb{R}^{NM \times NM}$ and a sufficiently large integer $T$, such that for any $t>T$, 
	\begin{align*}
\alpha(t) M_0\leq \beta(t)(\mathcal{L}\otimes I_M)+\alpha(t)D_H
< & I_{N\times M}.
	\end{align*}
\end{lemma}
\begin{proof}
	The proof is similar to Lemma 6 of \cite{kar2011convergence}.
\end{proof}

\begin{lemma}(Lemma 6, \cite{kar2013distributed})\label{lem_scalar}
Consider a scalar sequence $\{z(t)\}$  satisfying
	\begin{align*}
	z(t+1)= (1-r_1(t))z(t)+r_2(t),
	\end{align*}
	 with initial value $z(0)\geq 0$,
	where $r_1(t)=\frac{a_1}{(t+1)^{\delta_1}}$ and $r_2(t)= \frac{a_2}{(t+1)^{\delta_2}}$, with $0\leq r_1(t)\leq 1$, $a_1>0$,  $a_2>0$, $0\leq\delta_1\leq 1$, $\delta_2\geq 0$, and  $\delta_1<\delta_2$. Then
	\begin{itemize}
		\item if $\delta_1<1$, for all $\delta_0\in[0,\delta_2-\delta_1)$, we have
		  \begin{align}\label{eq_lim}
	\lim\limits_{t\rightarrow\infty}(t+1)^{\delta_0}z(t)=0.
		  \end{align}
		  \item if $\delta_1=1$ and $a_1>\delta_0$, (\ref{eq_lim}) holds.
	\end{itemize}
\end{lemma}

On the estimator (\ref{eq_estimator}), the asymptotic unbiasedness is studied in the following theorem.
\begin{theorem}(Asymptotically Unbiased)\label{thm_bias}
If $\rho_0>\tau_1-\tau_2$, the estimate sequence $\{x_{i}(t)\}$ by (\ref{eq_estimator}) is asymptotically unbiased for the true parameter $\theta$, i.e., $\lim\limits_{t\rightarrow\infty}\Expect\{x_{i}(t)\}=\theta$, $\forall i\in\mathcal{V}.$
\end{theorem}
\begin{proof}
According to (\ref{eq_estimator03}), we have
\begin{align}\label{eq_estimator2}
X(t+1)=&X(t)-\beta(t)(\mathcal{L}\otimes I_M)X(t)\nonumber\\
&+\alpha(t) D_H(\Theta-X(t))\nonumber\\
&+\alpha(t)\bar D_HV(t)\\
&+\beta(t)(\mathcal{A}\otimes I_M)(X(t_{k})-X(t))\nonumber.
\end{align}

Let $\tilde X(t)=X(t)-\Theta$ and $\bar X(t)=X(t_{k})-X(t)$. By $(\mathcal{L}\otimes I_M)\Theta=0$, we have
\begin{align}\label{eq_estimator3}
&\tilde X(t+1)\nonumber\\
=&\tilde X(t)-\beta(t)(\mathcal{L}\otimes I_M)\tilde X(t)-\alpha(t) D_H\tilde X(t)\nonumber\\
&+\alpha(t)\bar D_HV(t)+\beta(t)(\mathcal{A}\otimes I_M)\bar X(t)\\
=&\left(I_{M\times N}-\beta(t)(\mathcal{L}\otimes I_M)-\alpha(t) D_H\right)\tilde X(t)\nonumber\\
&+\alpha(t)\bar D_HV(t)+\beta(t)(\mathcal{A}\otimes I_M)\bar X(t).\nonumber
\end{align}
Taking expectation on both sides of (\ref{eq_estimator3}), we have
\begin{align}\label{eq_estimator4}
&\Expect\{\tilde X(t+1)\}\nonumber\\
=&\left(I_{M\times N}-\beta(t)(\mathcal{L}\otimes I_M)-\alpha(t) D_H\right)\Expect\{\tilde X(t)\}\nonumber\\
&+\beta(t)(\mathcal{A}\otimes I_M)\Expect\{\bar X(t)\}.
\end{align}
According to Lemma \ref{lem_posit}, there exists a  sufficiently large integer $T$, such that for any $t>T$,
	\begin{align*}
\alpha(t) M_0\leq \beta(t)(\mathcal{L}\otimes I_M)+\alpha(t) D_H
< & I_{N\times M}.
\end{align*}
Then, for $t>T$, taking norm operator on both sides of (\ref{eq_estimator4}) yields 
\begin{align}\label{eq_estimator5}
&\norm{\Expect\{\tilde X(t+1)\}}\nonumber\\
\leq&\norm{\left(I_{M\times N}-\beta(t)(\mathcal{L}\otimes I_M)-\alpha(t)  D_H\right)}\norm{\Expect\{\tilde X(t)\}}\nonumber\\
&+\beta(t)\norm{(\mathcal{A}\otimes I_M)}\norm{\Expect\{\bar X(t)\}}\\
\leq& (1-\alpha(t)m_0)\norm{\Expect\{\tilde X(t)\}}+\beta(t)MN\norm{\Expect\{\bar X(t)\}}\nonumber,
\end{align}
where $m_0=\lambda_{min}(M_0)$.

Recall $\beta(t)=\frac{b}{(t+1)^{\tau_2}}$, then there exists a constant scalar $m_1>0$, such that $\beta(t)\norm{\Expect\{\bar X(t)\}}\leq  \frac{m_1}{(t+1)^{\tau_2+\rho_0}}$.
As a result, from (\ref{eq_estimator5}), we have
\begin{align}\label{eq_estimator6}
&\norm{\Expect\{\tilde X(t+1)\}}\nonumber\\
\leq& (1-\alpha(t)m_0)\norm{\Expect\{\tilde X(t)\}}+\frac{MNm_1}{(t+1)^{\tau_2+\rho_0}}\\
=& (1-\frac{am_0}{(t+1)^{\tau_1}})\norm{\Expect\{\tilde X(t)\}}+\frac{MNm_1}{(t+1)^{\tau_2+\rho_0}}\nonumber.
\end{align}
Without losing generality, here we suppose $am_0>\rho_0+\tau_2-\tau_1$. Otherwise, we can obtain a sufficiently large $m_0$ by increasing $t$ and maintaining the value of $\frac{am_0}{(t+1)^{\tau_1}}.$ Due to $\rho_0>\tau_1-\tau_2$,  according to Lemma \ref{lem_scalar} and (\ref{eq_estimator6}), $\norm{\Expect\{\tilde X(t+1)\}}$ goes to zero as $t$ goes to infinity. 
\end{proof}
We can see from Theorem \ref{thm_bias} that the initial estimation biases of agents can be removed by the estimator (\ref{eq_estimator}) as time goes to infinity.

\begin{lemma}\label{lem_bound}
Under Assumptions \ref{ass_graph} - \ref{ass_steps}, if $\rho_0>0.5-\tau_2$,  there exists a finite random variable $R>0$, such that
\begin{align*}
\Pro\left(\sup_{i\in\mathbb{N}}\norm{X(t)}\leq R\right)=1.
\end{align*}
\end{lemma}
\begin{proof}
	Due to page limitation, the proof is omitted.
\end{proof}

To study the convergence of estimates in (\ref{eq_estimator}), first we  introduce a centralized estimator with strong consistency, i.e., the estimate sequence converges to the true parameter almost surely. Then, we prove the estimates of (\ref{eq_estimator}) can reach consensus, and the consensus value can asymptotically converge to the estimates of the centralized estimator. Thus, the strong consistency of estimates in (\ref{eq_estimator}) can be proved.
\begin{definition}(Centralized Linear Estimator)\label{df_centra}
	A centralized linear estimator $\{u(t)\}$ has the following form
	\begin{align}\label{eq_ut}
	u(t+1)=u(t)+\frac{\alpha_c(t)}{N} \sum_{i=1}^{N}H_i^T\left(y_i(t)-H_iu(t)\right),
	\end{align}
	where $\alpha(t)=\frac{a_c}{(t+1)^{\tau_c}}$ for some $a_c>0$ and $\tau_c\geq 0$. 
\end{definition}

\begin{lemma}\label{prop_centr}\cite{kar2011convergence}
	For the centralized linear estimator given in Definition \ref{df_centra}, we have the following results
	
	1) The estimate sequence $\{u(t)\}$ is of strong consistency w.r.t. $\theta$, i.e., 
	\begin{align}
	\mathbb{P}\left(\lim\limits_{t\rightarrow\infty}u(t)=\theta\right)=1.
	\end{align}
	
	2) Let $\alpha_c(t)=\frac{a_c}{(t+1)}$ with $a_c>\frac{N}{2\lambda_{min}(G)}$.
	Then the sequence $\{\sqrt{t+1}(u(t)-\theta)\}$ is asymptotically normal, i.e.,
	\begin{align*}
	\sqrt{t+1}(u(t)-\theta)\Rightarrow \mathcal{N}(0,S_c),
	\end{align*}
	where 
	\begin{align*}
	&S_c =\frac{a_c^2}{N^2}\int_{0}^{\infty}e^{\varSigma_1v}S_1e^{\varSigma_1^Tv}dv\\
	&\varSigma_1=-\frac{a_c}{N}G+\frac{1}{2}I_M\\
	&S_1=(\textbf{1}\otimes I_M)^T\bar D_{ H}R_v\bar D_{ H}^T(\textbf{1}\otimes I_M).
	\end{align*}
	
\end{lemma}

Define   $x_{avg}(t)=\frac{1}{N}(\textbf{1}_N\otimes  I_M)^TX(t)$.
In the following lemma, we provide conditions such that the estimates of agents reach consensus. 
\begin{lemma}\label{lem_consensus}
Let Assumptions \ref{ass_graph} - \ref{ass_steps} hold. Then,
 for any 	\begin{align*}
	0\leq\tau_0<\min\{\rho_0,\tau_1-\tau_2-\frac{1}{2+\epsilon_1}\}, 
	\end{align*}
	we have
	\begin{align*}
	\mathbb{P}\left(\lim\limits_{t\rightarrow\infty}(t+1)^{\tau_0}\norm{x_i(t)-x_{avg}(t)}=0\right)=1, \forall i\in\mathcal{V}.
	\end{align*}
\end{lemma}
\begin{proof}
		Due to page limitation, the proof is omitted.
\end{proof}

Next, we show that the consensus value, i.e., the average estimates, will converge to the estimates of the centralized estimator in (\ref{eq_ut}).
\begin{lemma}\label{lem_appro_centr}
	Let Assumptions \ref{ass_graph} - \ref{ass_steps} hold. 
	Suppose $\{u(t)\}$ is the centralized estimates given in Definition \ref{df_centra} with $\tau_c=\tau_1$, $a_c=a$.
	If $\rho_0>\tau_1-\tau_2,$
	then for any
		\begin{align*}
	0\leq\tau_0<\min\{\tau_1-\tau_2-\frac{1}{2+\epsilon_1},\rho_0+\tau_2-\tau_1\}, 
	\end{align*}
	we have	
\begin{align*}
	\mathbb{P}\left(\lim\limits_{t\rightarrow\infty}(t+1)^{\tau_0}\|x_{avg}(t)-u(t)\|=0\right)=1.
\end{align*}
\end{lemma}
\begin{proof}
		Due to page limitation, the proof is omitted.
\end{proof}
The strong consistency of   estimator (\ref{eq_estimator}) is provided in the following theorem.
\begin{theorem}(Strong Consistency)
	Consider the algorithm (\ref{eq_estimator}). Let Assumptions \ref{ass_graph} - \ref{ass_steps} hold.
	If $\rho_0>\tau_1-\tau_2$, the estimate sequence $\{x_i(t)\}$ is of strong consistency w.r.t. $\theta$, i.e.,
	\begin{align}
	\mathbb{P}\left(\lim\limits_{t\rightarrow\infty}x_i(t)=\theta\right)=1, \forall i\in\mathcal{V}.
	\end{align}
\end{theorem}
\begin{proof}
	According to Lemmas \ref{prop_centr} - \ref{lem_appro_centr}, taking $\tau_0=0$, the conclusion holds.
\end{proof}
In the next theorem, we provide the convergence speed that the estimates by (\ref{eq_estimator}) converge to the estimates of centralized estimator in (\ref{eq_ut}).
\begin{theorem}\label{thm_consistency}(Centralized Approximation) 
Let the algorithm (\ref{eq_estimator}) share the same parameter setting  as a centralized estimator $\{u(t)\}$ in that $ \tau_1=\tau_c$. Assume Assumptions \ref{ass_graph} - \ref{ass_steps} hold, and if $\tau_1=1$, further suppose 
\begin{align}
a>\frac{N\tau_0}{\lambda_{min}( G)}.
\end{align}
	Then, if $\rho_0>\tau_1-\tau_2$, for each  $\tau_0$ subject to
			\begin{align*}
	0\leq\tau_0<\min\{\tau_1-\tau_2-\frac{1}{2+\epsilon_1},\rho_0+\tau_2-\tau_1\}, 
	\end{align*}
we have
	\begin{align}\label{eq_cen_app}
	\mathbb{P}\left(\lim\limits_{t\rightarrow\infty}(t+1)^{\tau_0}\norm {x_i(t)-u(t)}=0\right)=1, \forall i\in\mathcal{V}.
	\end{align}
\end{theorem}
\begin{proof}
According to Lemma \ref{lem_consensus}, Lemma \ref{lem_appro_centr}, the conclusion holds.
\end{proof}
Communication frequency is essential to the research of event-triggered distributed estimation. In the following theorem, the triggering interval of the defined event in (\ref{eq_event}) is investigated in the sense of infinite time.
\begin{theorem}\label{thm_triggering}(Triggering Interval)
Let $t_{k}^i$ be $k$th triggering instant of agent $i$, Assumptions \ref{ass_graph} - \ref{ass_steps} hold and agents share the same threshold, i.e., $\rho_i=\rho_0$. If
\begin{align}\label{eq_triggering_upp}
\rho_0<\tau_1-\frac{1}{2+\epsilon_1},
\end{align}
then for each agent, 
 the time interval between two successive triggered instants goes to infinity, i.e., 
\begin{align}\label{eq_limt}
\Pro\left(\lim\limits_{k\rightarrow\infty}(t_{k+1}^i-t_{k}^i)=\infty\right)=1, \forall i\in\mathcal{V}.
\end{align} 
\end{theorem}
	\begin{proof}
Note that $t_{k}^i$ is $k$th triggering instant of agent $i$, then we focus on analyzing the 
time interval length of $t_{k+1}^i-t_{k}^i$ in the following.

According to (\ref{eq_estimator02}), for $t\geq t_{k}^i$, we have 
\begin{align}\label{eq_estimator01}
&x_{i}(t+1)-x_{i}(t_{k}^i)\\
=&x_{i}(t)-x_{i}(t_{k}^i)+\beta(t)\sum_{j\in\mathcal{N}_{i}}(x_{j}(t)-x_{i}(t))\nonumber\\
&+\alpha(t)  H_{i}^T( y_{i}(t)-H_{i}x_{i}(t))\nonumber\\
&+\beta(t)\sum_{j\in\mathcal{N}_{i}}(x_{j}(t_{k}^j)-x_{j}(t))\nonumber.
\end{align}
Taking norm operator on both sides of (\ref{eq_estimator01}) yields
\begin{align}\label{eq_estimator012}
&\norm{x_{i}(t+1)-x_{i}(t_{k}^i)}\\
\leq&\norm{x_{i}(t)-x_{i}(t_{k}^i)}+\beta(t)\norm{\sum_{j\in\mathcal{N}_{i}}(x_{j}(t)-x_{i}(t))}\nonumber\\
&+\alpha(t) \norm{ H_{i}^T( y_{i}(t)-H_{i}x_{i}(t))}\nonumber\\
&+\beta(t)\norm{\sum_{j\in\mathcal{N}_{i}}(x_{j}(t_{k}^j)-x_{j}(t))}\nonumber.
\end{align}
According to Lemma \ref{lem_consensus}, for $	0\leq\tau_0<\min\{\rho_0,\tau_1-\tau_2-\frac{1}{2+\epsilon_1}\}$, we have
	\begin{align*}
\mathbb{P}\left(\lim\limits_{t\rightarrow\infty}(t+1)^{\tau_0}\norm{x_i(t)-x_{j}(t)}=0\right)=1, \forall i,j\in\mathcal{V},
\end{align*}
Then there exits a scalar $c_3>0$, such that
\begin{align}\label{eq_c3}
\norm{\sum_{j\in\mathcal{N}_{i}}(x_{j}(t)-x_{i}(t))}\leq c_3 (t+1)^{-\tau_0}. 
\end{align}
By Lemma \ref{lem_bound} and Assumption \ref{ass_noise},  there exists a scalar $c_4>0$ and a sufficiently small $\delta>0$, such that
\begin{align}\label{eq_c4}
\norm{ H_{i}^T( y_{i}(t)-H_{i}x_{i}(t))}\leq c_4(t+1)^{\frac{1}{2+\epsilon_1}+\delta}.
\end{align}
According to Lemma \ref{eq_error_trans}, there exists a scalar $c_5>0$ such that 
\begin{align}\label{eq_c5}
\norm{\sum_{j\in\mathcal{N}_{i}}(x_{j}(t_{k}^j)-x_{j}(t))}\leq c_5(t+1)^{-\rho_0}.
\end{align}
Taking (\ref{eq_c3}), (\ref{eq_c4}) and (\ref{eq_c5}) into (\ref{eq_estimator012}), we have
\begin{align}\label{eq_estimator002}
&\norm{x_{i}(t+1)-x_i(t_{k}^i)}\\
\leq&\norm{x_{i}(t)-x_i(t_{k}^i)}+\beta(t)c_3 (t+1)^{-\tau_0}\nonumber\\
&+\alpha(t) c_4(t+1)^{\frac{1}{2+\epsilon_1}+\delta}+\beta(t)c_5(t+1)^{-\rho_0}\nonumber\\
\leq &\norm{x_{i}(t)-x_i(t_{k}^i)}+c_3\frac{1}{(t+1)^{\tau_0+\tau_2}}\nonumber\\\
&+ c_4\frac{1}{(t+1)^{\tau_1-\frac{1}{2+\epsilon_1}-\delta}}+c_5\frac{1}{(t+1)^{\rho_0+\tau_2}}.
\end{align}
Considering  $	0\leq\tau_0<\min\{\rho_0,\tau_1-\tau_2-\frac{1}{2+\epsilon_1}\}$, by choosing a sufficiently small $\delta$, we have
$\tau_0+\tau_2\leq \min\{\tau_1-\frac{1}{2+\epsilon_1}-\delta,\rho_0+\tau_2\}$. 
Denote $t_{k+1}^i=t_{k}^i+L_{k}^i$.
Then, there exists a sufficiently large integer $T_6$ and a scalar $c_6$, such that for
$t_{k}^i>T_6$, 
\begin{align*}
&\norm{x_i(t_{k+1}^i)-x_i(t_{k}^i)}\\
=&\norm{x_i(t_{k}^i+L_{k}^i)-x_i(t_{k}^i)}\\
\leq&\norm{x_i(t_{k}^i+L_{k}^i-1)-x_i(t_{k}^i)}+c_6 \frac{1}{(t_{k}^i+L_{k}^i)^{\tau_0+\tau_2}}\\
&\qquad\vdots\\
\leq&c_6\sum_{s=t_{k}^i}^{t_{k}^i+L_{k}^i-1} \frac{1}{(s+1)^{\tau_0+\tau_2}}\\
\leq&\frac{c_6}{(t_{k}^i+1)^{\tau_0+\tau_2}}L_{k}^i.
\end{align*}
A necessary condition to guarantee that the event in (\ref{eq_event}) is triggered for agent $i$ is 
\begin{align}\label{eq_equli}
&\frac{c_6}{(t_{k}^i+1)^{\tau_0+\tau_2}}L_{k}^i>\frac{1}{ (t_{k+1}^i+1)^{\rho_0}}\nonumber\\
\Longleftrightarrow&\frac{c_6}{(t_{k}^i+1)^{\tau_0+\tau_2}}L_{k}^i>\frac{1}{ (t_{k}^i+1+L_{k}^i)^{\rho_0}}\nonumber\\
\Longleftrightarrow&\frac{(t_{k}^i+1+L_{k}^i)^{\rho_0}}{(t_{k}^i+1)^{\tau_0+\tau_2}}L_{k}^i>\frac{\rho_0}{c_6}.
\end{align}
Due to $0\leq\tau_0<\min\{\rho_0,\tau_1-\tau_2-\frac{1}{2+\epsilon_1}\}$, there exists a scalar $\bar\delta>0$, such that $\tau_0=\min\{\rho_0,\tau_1-\tau_2-\frac{1}{2+\epsilon_1}\}-\bar\delta$.
Then, 
\begin{align*}
\tau_0+\tau_2=\min\{\rho_0+\tau_2,\tau_1-\frac{1}{2+\epsilon_1}\}-\bar\delta
\end{align*}
Recall the condition (\ref{eq_triggering_upp}), and let $\bar\delta$ go to zero, then $\tau_0+\tau_2>\rho_0.$
%
To make sure the satisfaction of (\ref{eq_limt}), we need to show that $L_{k}^i$ goes to infinity when $t_{k}^i$ goes to infinity. By contradiction, we suppose that there is an integer $\bar L^i$, such that $L_{k}^i\leq \bar L^i$, $\forall k\in\mathbb{N}$. A necessary condition of (\ref{eq_equli}) is 
\begin{align}
\frac{(t_{k}^i+1+\bar L^i)^{\rho_0}}{(t_{k}^i+1)^{\tau_0+\tau_2}}>\frac{\rho_0}{c_6\bar L^i},
\end{align}
which however cannot be satisfied as $t_{k}^i$ is very large due to $\tau_0+\tau_2>\rho_0$. Therefore, $L_{k}^i$ goes to infinity as $t_{k}^i$ goes to infinity.
\end{proof}

\section{Numerical Simulation}
In this section, we provide a numerical simulation to testify the effectiveness of distributed estimator based on event-triggered communication scheme proposed in this paper.

Consider an undirected network with four agents. The adjacency matrix of the network is $\mathcal{A}=\left(\begin{smallmatrix}
0& 1&0&1\\
1& 0&1&0\\
0& 1&0&1\\
1& 0&1&0
\end{smallmatrix}\right)$.
 The true parameter vector is supposed to be $\theta=[-1,2]^T.$ The observation matrices and the initial parameter estimates of these agents have the following forms
\begin{align*}
&H_1=[1,0]^T, H_2=[0,1]^T\\
&H_3=[1,1]^T, H_4=[1,2]^T\\
&x_{1}(0)=[10,20]^T,x_{2}(0)=[10,-10]^T\\
&x_{3}(0)=[10,-20]^T,x_{2}(0)=[20,-10]^T.
\end{align*}
We consider the time sequence $t=0,1,\dots,10000$. Let $\alpha(t)=\frac{1}{(t+1)^{0.7}}$, $\beta(t)=\frac{1}{(t+1)^{0.5}}$ and $\rho_{i}(t)=\frac{1}{(t+1)^{0.6}}$, for $i=1,2,3,4.$ The noises of each agent are supposed to be $i.i.d$ and Gaussian. The noises of agents are spatially independent. The distribution of measurement noises is mean zero and variance $0.01$. 

Under the above setting, by employing the distributed estimator (\ref{eq_estimator}) with triggering scheme (\ref{eq_trigger_scheme}) and the centralized estimator \ref{eq_ut}, we obtain simulation results in Fig. \ref{fig:estimates}, Fig. \ref{fig:error} and Fig. \ref{fig:rate}.
We see from Fig. \ref{fig:estimates} that the average estimates are asymptotically convergent to the true parameters of the system. By Fig. \ref{fig:error}, the consistency of the estimator for each agent is shown. Besides, we see that the centralized estimator has faster convergence speed, since it utilizes all measurements. The triggering time instants satisfying the  triggering scheme (\ref{eq_trigger_scheme}) during the whole estimation process is plotted in Fig. \ref{fig:rate} with communication rate\footnote{Communication rate is the ratio  of whole triggering time instants over the whole time-driven communication time instants} $1.175\%$. Thus, the communication frequency of the agents has been tremendously reduced with guaranteed convergence properties. 

\begin{figure}[htp]
	\centering
	\includegraphics[scale=0.6]{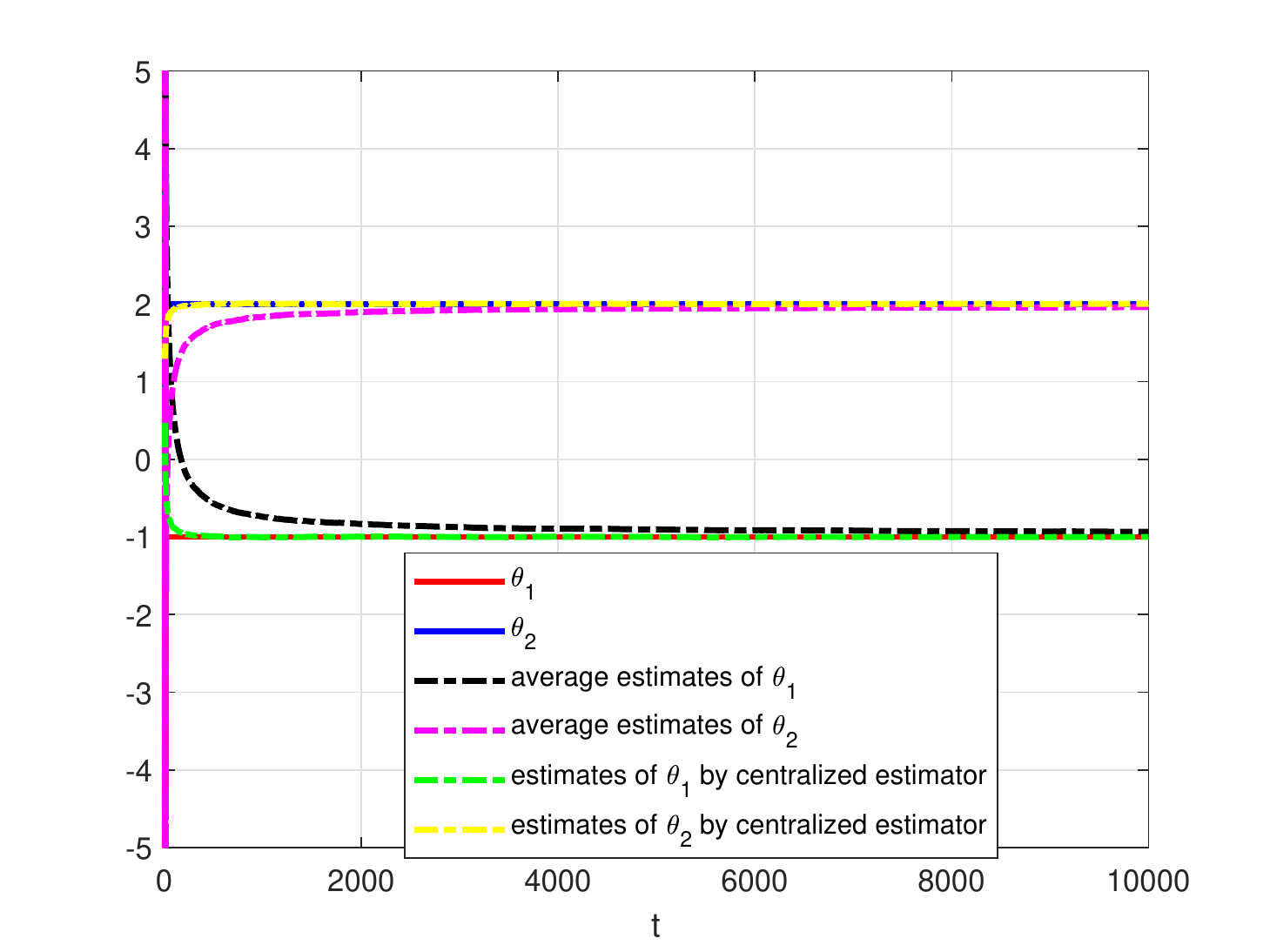}
	\caption {Parameter estimates of agents (average) and data center (centralized)}
	\label{fig:estimates}
\end{figure}

\begin{figure}[htp]
	\centering
	\includegraphics[scale=0.6]{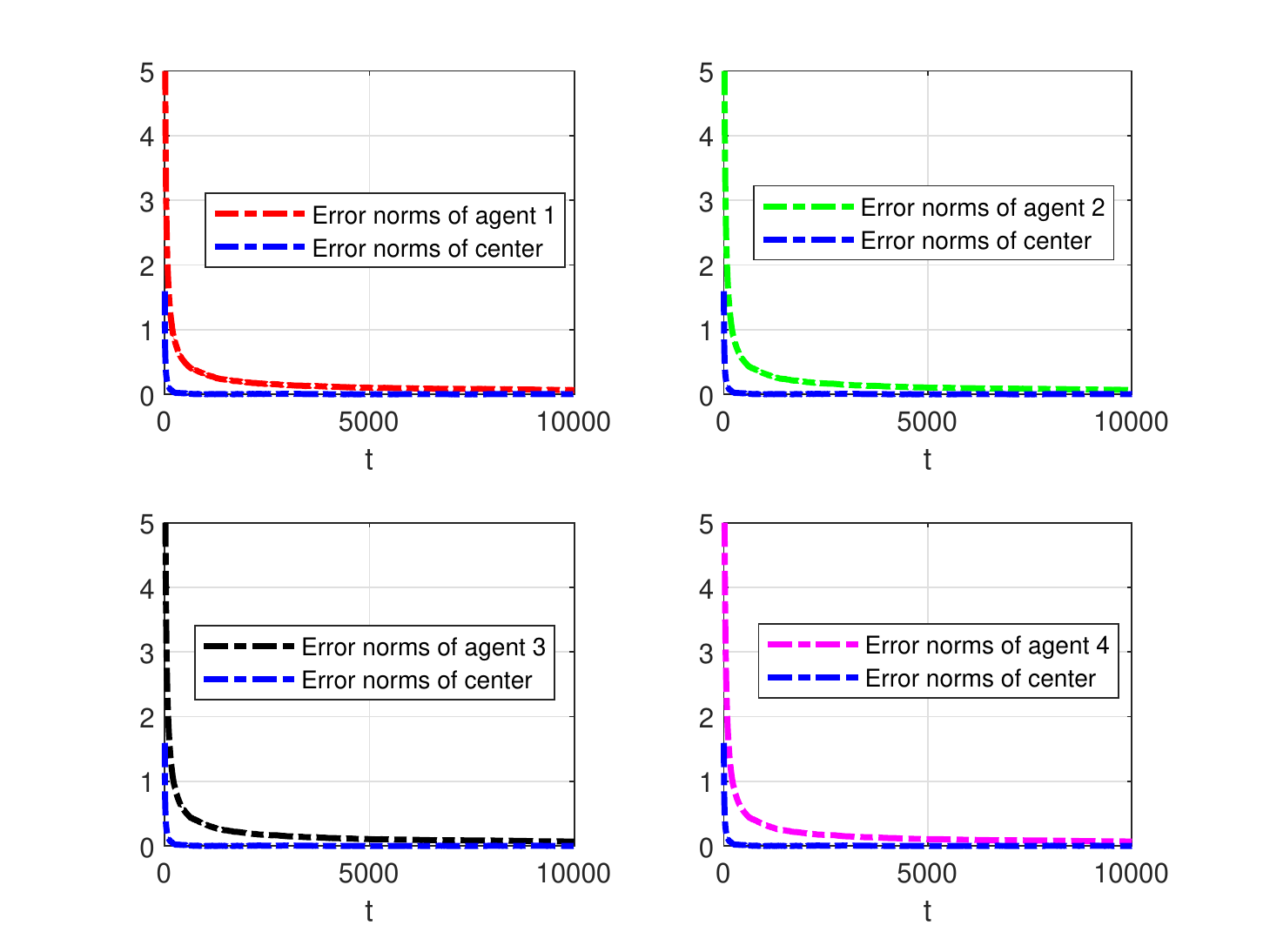}
	\caption {Estimation error norms of agents and data center}
	\label{fig:error}
\end{figure}

\begin{figure}[htp]
	\centering
	\includegraphics[scale=0.6]{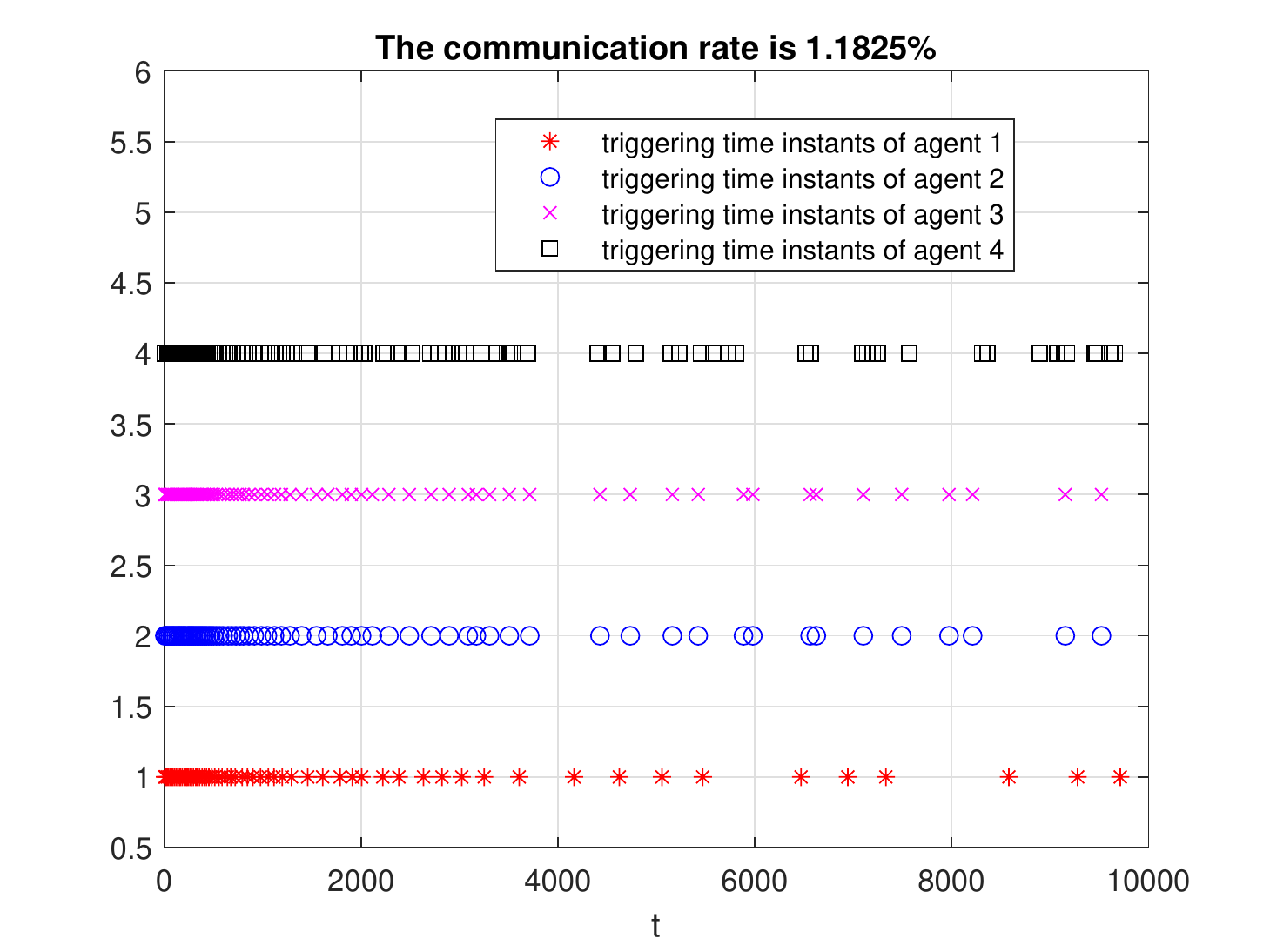}
	\caption {Triggering time instants and communication rate}
	\label{fig:rate}
\end{figure}

\section{Conclusion}
In this paper, a distributed parameter estimation problem with intermittent communications was studied.
First, we proposed an event-triggered communication scheme for each agent, by comparing a decaying threshold with the difference between the current estimate and the latest one sent out to   neighboring agents. Then, we analyze some main estimation properties including asymptotic  unbiasedness and strong consistency. We also showed that, with probability one, for every agent the time interval between two successive triggered instants goes to infinity as time goes to infinity.

%

\ifCLASSOPTIONcaptionsoff
  \newpage
\fi

\bibliography{C:/work_at_kth/All_references}

\begin{thebibliography}{10}

\bibitem{rad2010distributed}
K.~R. Rad and A.~Tahbaz-Salehi, ``Distributed parameter estimation in
  networks,'' in {\em IEEE Conference on Decision and Control}, pp.~5050--5055,
  2010.

\bibitem{kar2011convergence}
S.~Kar and J.~M. Moura, ``Convergence rate analysis of distributed gossip
  (linear parameter) estimation: Fundamental limits and tradeoffs,'' {\em IEEE
  Journal of Selected Topics in Signal Processing}, vol.~5, no.~4,
  pp.~674--690, 2011.

\bibitem{kar2013distributed}
S.~Kar, J.~M. Moura, and H.~V. Poor, ``Distributed linear parameter estimation:
  Asymptotically efficient adaptive strategies,'' {\em SIAM Journal on Control
  and Optimization}, vol.~51, no.~3, pp.~2200--2229, 2013.

\bibitem{cattivelli2010diffusion}
F.~S. Cattivelli and A.~H. Sayed, ``Diffusion strategies for distributed
  \protect{Kalman} filtering and smoothing,'' {\em IEEE Transactions on
  Automatic Control}, vol.~55, no.~9, pp.~2069--2084, 2010.

\bibitem{zhang2012distributed}
Q.~Zhang and J.-F. Zhang, ``Distributed parameter estimation over unreliable
  networks with markovian switching topologies,'' {\em IEEE Transactions on
  Automatic Control}, vol.~57, no.~10, pp.~2545--2560, 2012.

\bibitem{kar2012distributed}
S.~Kar, J.~M. Moura, and K.~Ramanan, ``Distributed parameter estimation in
  sensor networks: Nonlinear observation models and imperfect communication,''
  {\em IEEE Transactions on Information Theory}, vol.~58, no.~6,
  pp.~3575--3605, 2012.

\bibitem{you2013asymptotically}
K.~You, L.~Xie, and S.~Song, ``Asymptotically optimal parameter estimation with
  scheduled measurements,'' {\em IEEE Transactions on Signal Processing},
  vol.~61, no.~14, pp.~3521--3531, 2013.

\bibitem{shi2014event}
D.~Shi, T.~Chen, and L.~Shi, ``Event-triggered maximum likelihood state
  estimation,'' {\em Automatica}, vol.~50, no.~1, pp.~247--254, 2014.

\bibitem{mo2012kalman}
Y.~Mo and B.~Sinopoli, ``\protect{Kalman} filtering with intermittent
  observations: Tail distribution and critical value,'' {\em IEEE Transactions
  on Automatic Control}, vol.~57, no.~3, pp.~677--689, 2012.

\bibitem{sinopoli2004kalman}
B.~Sinopoli, L.~Schenato, M.~Franceschetti, K.~Poolla, M.~I. Jordan, and S.~S.
  Sastry, ``\protect{Kalman} filtering with intermittent observations,'' {\em
  IEEE Transactions on Automatic Control}, vol.~49, no.~9, pp.~1453--1464,
  2004.

\bibitem{han2015optimal}
D.~Han, K.~You, L.~Xie, J.~Wu, and L.~Shi, ``Optimal parameter estimation under
  controlled communication over sensor networks.,'' {\em IEEE Trans. Signal
  Processing}, vol.~63, no.~24, pp.~6473--6485, 2015.

\bibitem{han2015stochastic}
D.~Han, Y.~Mo, J.~Wu, S.~Weerakkody, B.~Sinopoli, and L.~Shi, ``Stochastic
  event-triggered sensor schedule for remote state estimation,'' {\em IEEE
  Transactions on Automatic Control}, vol.~60, no.~10, pp.~2661--2675, 2015.

\bibitem{weimer2012distributed}
J.~Weimer, J.~Ara{\'u}jo, and K.~H. Johansson, ``Distributed event-triggered
  estimation in networked systems,'' {\em IFAC Proceedings Volumes}, vol.~45,
  no.~9, pp.~178--185, 2012.

\bibitem{He2017On}
X.~He, C.~Hu, W.~Xue, and H.~Fang, ``On event-based distributed
  \protect{Kalman} filter with information matrix triggers,'' in {\em IFAC
  World Congress}, pp.~14873--14878, 2017.

\bibitem{battistelli2018distributed}
G.~Battistelli, L.~Chisci, and D.~Selvi, ``A distributed \protect{Kalman}
  filter with event-triggered communication and guaranteed stability,'' {\em
  Automatica}, vol.~93, pp.~75--82, 2018.

\bibitem{he2017distributed}
X.~He, C.~Hu, Y.~Hong, L.~Shi, and H.~Fang, ``Distributed \protect{Kalman}
  filters with state equality constraints: Time-based and event-triggered
  communications,'' {\em arXiv preprint arXiv:1711.05010}, 2017.

\bibitem{Mesbahi2010Graph}
M.~Mehran and E.~Magnus, {\em Graph theoretic methods in multiagent networks}.
\newblock Princeton University Press, 2010.

\end{thebibliography}
\bibliographystyle{ieeetr}

%
%

%
%
%

\end{document}